\documentclass[journal,twoside,web]{ieeecolor}
\usepackage{generic}
\usepackage{cite}
\usepackage{amssymb,amsfonts}
\usepackage{algorithmic}
\usepackage{graphicx}
\usepackage{textcomp}
\def\BibTeX{{\rm B\kern-.05em{\sc i\kern-.025em b}\kern-.08em
    T\kern-.1667em\lower.7ex\hbox{E}\kern-.125emX}}
\usepackage[dvipsnames]{xcolor}
\usepackage{color, soul}
\sethlcolor{red}
\newcommand{\stirlingii}{\genfrac{\{}{\}}{0pt}{}}
\newcommand{\stirlingi}{\genfrac{[}{]}{0pt}{}}

\usepackage{textcomp}
\def\BibTeX{{\rm B\kern-.05em{\sc i\kern-.025em b}\kern-.08em
    T\kern-.1667em\lower.7ex\hbox{E}\kern-.125emX}}
\usepackage{float}
\usepackage[tbtags]{amsmath}
\usepackage[tbtags]{mathtools}
\usepackage{accents}
\usepackage{multicol}
\usepackage[ruled,vlined]{algorithm2e}
\usepackage{cite}
\usepackage{pgfplots}

\usepackage{array,multirow}
\usepackage{hhline}
\usepackage{arydshln}

\usepackage[hidelinks=true]{hyperref}

\newtheorem{theorem}{Theorem}

\newtheorem{lemma}{Lemma}
\newtheorem{example}{Example}
\newtheorem{assumption}{Assumption}

\newtheorem{remark}{Remark}

\newtheorem{definition}{Definition}

\begin{document}
\title{Prescribed-Time Control with Linear Decay for Nonlinear Systems}
\author{Amir Shakouri, \IEEEmembership{Member, IEEE}, and Nima Assadian, \IEEEmembership{Member, IEEE}
\thanks{The authors are with the Department of Aerospace Engineering, Sharif University of Technology, Tehran, Iran (e-mail: \href{mailto:a_shakouri@outlook.com}{a$\_$shakouri@outlook.com}; \href{mailto:assadian@sharif.edu}{assadian@sharif.edu}).}}

\maketitle

\pagestyle{empty}
\thispagestyle{empty}

\begin{abstract}
In this letter, a new notion of stability is introduced, which is called \textit{triangular stability}. A system is called triangularly stable if the norm of its state vector is bounded by a decreasing linear function of time such that its intersection point with the time axis can be arbitrarily commanded by the user. Triangular stability implies prescribed-time stability, which means that the nonlinear system is converged to zero equilibrium at an arbitrary finite time. A prescribed-time controller with guaranteed triangular stability is developed for normal form nonlinear systems with uncertain input gain, which is able to reject the disturbances and unmodeled dynamics. Numerical simulations are carried out to visualize the results for second and fourth-order systems. 
\end{abstract}

\section{Introduction}
\label{sec:I}
\IEEEPARstart{P}{rescribed-time} control methods are time-varying techniques for stabilizing nonlinear systems in an exact arbitrary finite time. Comparing to finite-time \cite{bhat2000finite} and fixed-time control \cite{polyakov2011nonlinear} methods, the prescribed-time control has a short history in the literature. In finite-time control, the system converges within a finite time that is unknown or depends on the initial conditions of the system. Fixed-time control provides a solution by which an upper bound, independent of initial conditions, is obtainable for the time of convergence. On the other hand, a prescribed-time controller (PTC) forces the system to automatically  converge at the intended time. The disturbance rejection ability of the PTCs, besides their smooth chattering-free behavior, makes this class of controllers much useful for many engineering applications. Safety and precision of uncertain systems, especially when they are subject to time-varying environments or cooperating with other dynamic systems, can be much improved if the user can directly command the stabilization time to the system.

The use of time-varying approaches to stabilize a nonlinear system in a prescribed finite time initially proposed by Song et al. \cite{song2017time}. In the prescribed-time control methods, a mapping from the infinite time scale onto an arbitrary finite time scale is the key idea to achieve a time-varying controller. The behavior of time-varying methods under non-vanishing uncertainties is studied by Wang et al. \cite{wang2019general}. A similar approach called the generalized time transformation method is proposed by Tran et al. \cite{tran2020finite} and studied for multiple systems by Arabi et al. \cite{arabi2019robustness}. Krishnamurthy et al. \cite{krishnamurthy2020dynamic} proposed a prescribed-time controller for systems with matched uncertainties, and the prescribed-time control of systems with uncertain input gains has been studied in \cite{krishnamurthy2020prescribed}. Autonomous methods for predefined-time control are analyzed in \cite{jimenez2020lyapunov} and a combination of autonomous and time-varying methods are investigated in \cite{gomez2020design}. In these works, the disturbances need to be globally bounded by a known constant. Thus, their applications are limited, and they cannot be used for unknown systems or known systems subject to state-dependent unmodeled dynamics. Moreover, autonomous predefined-time controllers may suffer from chattering under disturbances near equilibrium.

In this letter, a new notion of nonasymptotic stability is defined for nonlinear systems, called \textit{triangular stability}, which implies that the system solution is bounded by a decreasing linear function of time. We have proposed a PTC scheme for disturbed normal form systems with uncertain input gain, by which the closed-loop response is globally triangularly stable (or attractive), a strict stability condition that is not possible by the current state-of-the-art PTCs \cite{song2017time,wang2019general}.

\section{Preliminaries}

This section is devoted to introducing the basic notations, defining the nonasymptotic notions of stability, and formulating the Stirling numbers and matrices.

\subsection{Notations}

Let $\mathbb{R}^{m\times n}$ denote the space of $m \times n$ real matrices and $\mathbb{R}^n$ denotes the space of $n$-dimensional real vectors. The $n$-dimensional identity matrices is denoted by $\mathbb{I}_n$. The $i$th entry of vector $r\in \mathbb{R}^n$ is referred to by $r_i$ and the $ij$th entry of matrix $R$ is shown by $R(i,j)$. For matrix $R\in\mathbb{R}^{n\times n}$ we denote by $R^{-1}$ its inverse (if it exists). An inverse function is denoted by $f^{-1}(\cdot)$ for function $f(\cdot)$ (if the inverse exists). The symbol $\|\cdot\|$ denotes the $2$-norm for vectors and matrices. The Hadamard and Kronecker products are denoted by $\circ$ and $\otimes$, respectively. The maximum and minimum eigenvalues of a matrix $R$ is referred to by $\lambda_{\mathrm{max}}(R)$ and $\lambda_{\mathrm{min}}(R)$, respectively. The $i$th derivative of a function $f$ with respect to its argument is shown by $f^{(i)}$. The uniform distribution of a random variable between $a$ and $b$ is denoted by $\mathcal{U}(a,b)$.

\subsection{Stability Notions}

For a general $n$-dimensional nonlinear system as
\begin{equation}
\label{eq:1}
\dot{x}=a(x,t):\hspace{4mm}a(0,t)=0,\forall t\geq0
\end{equation}
three conventional notions of stability, known as \textit{Lyapunov stability}, \textit{global asymptotic stability}, and \textit{global exponential stability}, are frequently used in the analysis of control systems. In view of nonasymptotic techniques, in addition to the Lyapunov stability, the concepts of \textit{global finite-time stability} and \textit{global fixed-time stability} are defined as follows:
\begin{definition}[\cite{polyakov2011nonlinear}]
\label{def:FTS}
Let $x_0$ denote the state vector of \eqref{eq:1} at $t=0$. Then, the zero equilibrium of system \eqref{eq:1} is called
\begin{enumerate}
\item \textit{globally finite-time stable}, if it is globally asymptotically stable and there exists a settling time function $t_{FTC}(x_0):\mathbb{R}^n\rightarrow(0,\infty)$ such that for all $t\in[t_{FTC}(x_0),\infty)$ we have $x(t)=0$. 
\item \textit{globally fixed-time stable}, if it is globally finite-time stable and there exists $\bar{t}_{FTC}>0$ such that $t_{FTC}(x_0)<\bar{t}_{FTC}$ for all $x_0$, i.e., it is finite-time stable and an upper bound is known for the convergence time. 
\end{enumerate}
\end{definition}

In addition to the above definitions regarding system \eqref{eq:1}, for a system as follows:
\begin{equation}
\label{eq:1p}
\dot{x}=a(x,t,\tau):\hspace{4mm}a(0,t,\tau)=0,\forall t\geq0,\forall\tau>0
\end{equation}
where $\tau\in(0,\infty)$ is a user-defined parameter, the \textit{global prescribed-time stability} can be defined as:

\begin{definition}
\label{def:PTS}
Let $x_0$ denote the state vector of \eqref{eq:1p} at $t=0$. Then, the zero equilibrium of system \eqref{eq:1p} is called \textit{globally prescribed-time stable}, if it is globally finite-time stable and for every $\tau>0$ we have $t_{FTC}(x_0)<\tau$, i.e., it is finite-time stable and the convergence time can be arbitrarily specified. 
\end{definition}

In this letter, we define a new notion of stability that is inspired by the exponential stability for infinite-time systems:

\begin{definition}
\label{def:TS}
Let $x_0$ denote the state vector of \eqref{eq:1p} at $t=0$. Then, the zero equilibrium of system \eqref{eq:1p} is called \textit{globally triagularly stable}, if there exists $\sigma>1$ such that for every $x_0\in\mathbb{R}^n$ and $\tau>0$ we have:
\begin{equation}
\label{eq:2}
\|x(t)\|\leq\sigma\|x_0\|\Lambda(t/\tau),\hspace{4mm}\forall t\in[0,\infty),
\end{equation}
where $\Lambda(\cdot):\mathbb{R}\rightarrow[0,1]$ is the triangular function defined as (see Fig. \ref{fig:1})
\begin{equation}
\label{eq:3}
\Lambda(t/\tau)\coloneqq\max\{1-t/\tau,0\}=\left\{\begin{array}{lcl}
1-t/\tau & \mathrm{if} & t<\tau \\
0     & \mathrm{if} & t\geq\tau
\end{array}\right.
\end{equation}
\end{definition}

\begin{figure}[!h]
\centering\includegraphics[width=0.7\linewidth]{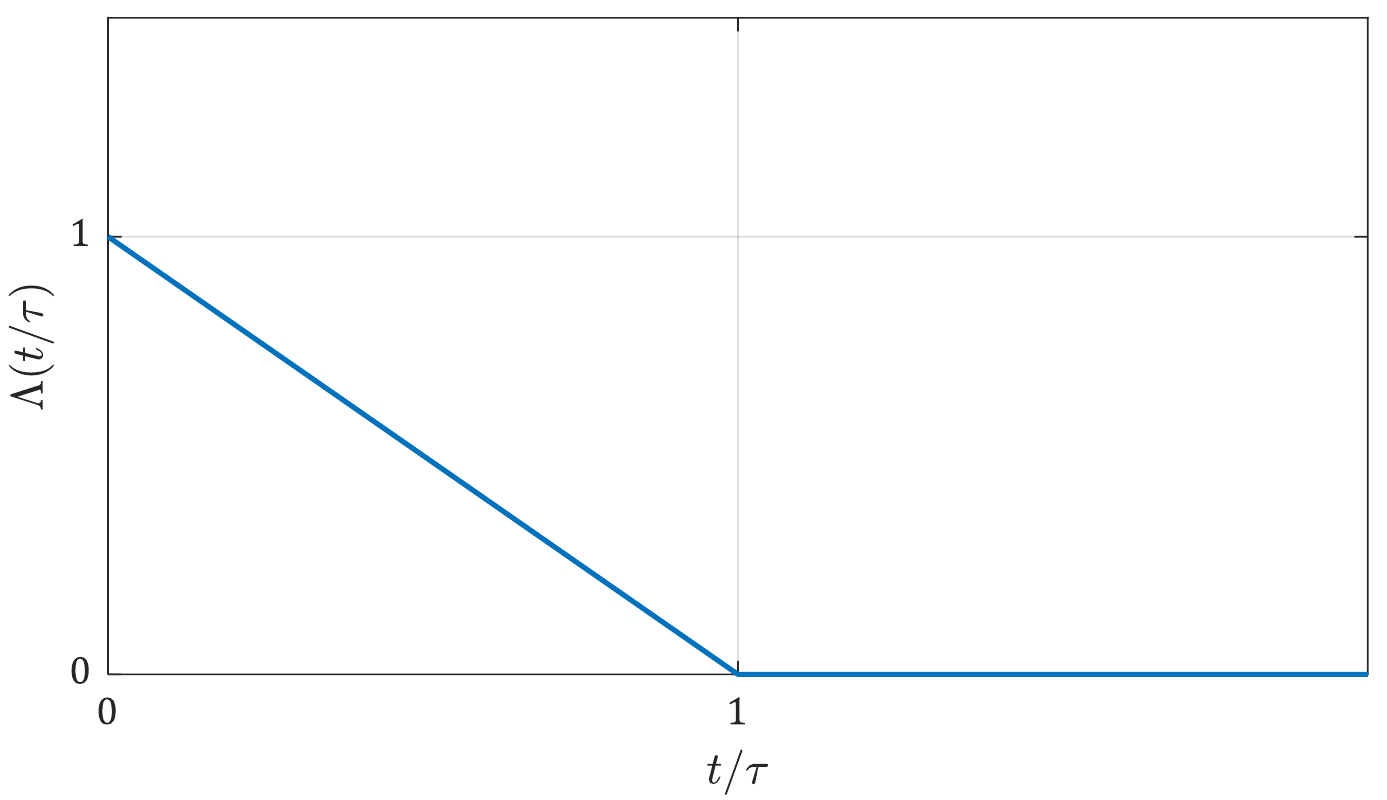}
\caption{Triangular function.}
\label{fig:1}
\end{figure}

Moreover, we define the triangular attractivity as follows with a less conservative condition:

\begin{definition}
\label{def:TA}
Let $x_0$ denote the state vector of \eqref{eq:1p} at $t=0$. Then, the zero equilibrium of system \eqref{eq:1p} is called \textit{globally triagularly attractive}, if there exists $\varsigma>0$ and $t_0<\tau$ such that for every $x_0\in\mathbb{R}^n$ and $\tau>0$ we have:
\begin{equation}
\label{eq:2}
\|x(t)\|\leq\varsigma\Lambda(t/\tau),\hspace{4mm}\forall t\in[t_0,\tau),
\end{equation}
\end{definition}

A triangularly stable nonlinear system is bounded by a linear function with a negative slope such that its intersection point with the time axis, $t=\tau$, is specifiable. It can be seen that triangular stability implies prescribed-time stability, prescribed-time stability implies fixed-time stability, and fixed-time stability implies finite-time stability, but the vice versa does not necessarily hold (see Fig. \ref{fig:ven}).

\begin{remark}
\label{rem:1}
Note that the difference between the fixed-time and prescribed-time notions of stability is the fact that in fixed-time stable systems, an upper bound $\bar{t}_{FTC}$ exists. This $\bar{t}_{FTC}$ cannot be simply commanded to the system since it can be a function of the system model and perturbations. Determining fixed-time control parameters by which $\bar{t}_{FTC}$ be lower than the desired value, needs considerable computational burden and restricting assumptions, if not impossible. However, in the prescribed-time method, the user specifies the convergence time $\tau$ as an input parameter to the control system, without concerning about the system dynamics. 
\end{remark}

\begin{figure}[!h]
\centering\includegraphics[width=1\linewidth]{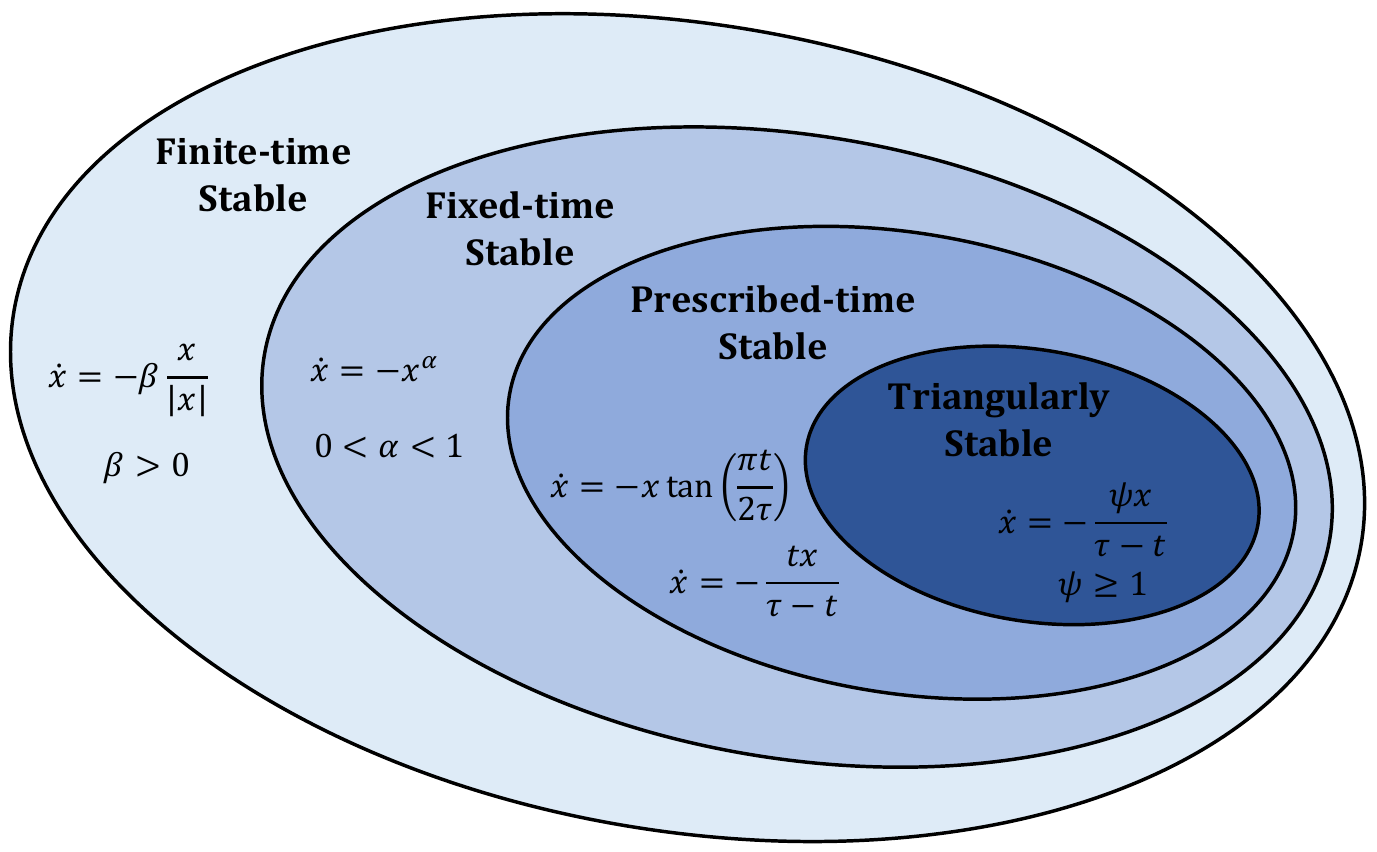}
\caption{Venn diagram representation of the nonasymptotic notions of stability with some examples.}
\label{fig:ven}
\end{figure}

\subsection{Stirling Numbers} 

Stirling numbers of the first kind are shown by $\stirlingi{n}{k}$ and they can be obtained from the following recursive formula:
\begin{equation}
\label{eq:A:2}
\stirlingi{n+1}{k}=n\stirlingi{n}{k}+\stirlingi{n}{k-1}
\end{equation}
for $k>0$ knowing that $\stirlingi{0}{0}=1$ and $\stirlingi{n}{0}=\stirlingi{0}{n}=1$. For Stirling numbers of the first kind one can verify that $\sum_{k=0}^n\stirlingi{n}{k}=n!$.

Stirling numbers of the second kind are associated with two indices $n$ and $k$ which are shown by $\stirlingii{n}{k}$ and can be obtained by the following explicit formula:
\begin{equation}
\label{eq:A:1}
\stirlingii{n}{k}=\frac{1}{k!}\sum_{i=1}^k(-1)^i\binom{k}{i}(k-i)^n
\end{equation}
The $n$th Bell number is $B_n=\sum_{k=0}^n\stirlingii{n}{k}$. We define the first and second kinds of Stirling matrices, denoted by $s_n\in\mathbb{R}^{n\times n}$ and $S_n\in\mathbb{R}^{n\times n}$, respectively, as follows \cite{comtet2012advanced}:
\begin{equation}
\label{eq:A:4}
s_n(i,j)=\left\{\begin{array}{lcl}
\stirlingi{i-1}{j-1} & \mathrm{if} & i\geq j \\
0 & \mathrm{if} &  \mathrm{otherwise}
\end{array}\right.
\end{equation}
\begin{equation}
\label{eq:A:3}
S_n(i,j)=\left\{\begin{array}{lcl}
\stirlingii{i-1}{j-1} & \mathrm{if} &  i\geq j \\
0 & \mathrm{if} &  \mathrm{otherwise}
\end{array}\right.
\end{equation}

\section{Main Results}
\label{sec:main}

We consider systems in the Byrnes-Isidori normal form \cite{byrnes1985global} with uncertain input gain as follows:
\begin{equation}
\label{eq:4}
\left\{\begin{array}{lcl}
\dot{x}_i & = & x_{i+1}; \hspace{2mm} i=1:n-1  \\
\dot{x}_n & = & f(x,u,t)+\gamma g(t)u
\end{array}\right.
\end{equation}
where $x=[x_1,\cdots,x_n]^T\in\mathbb{R}^n$ is the state vector, $u\in\mathbb{R}$ is the control input, $f(\cdot,\cdot,\cdot):\mathbb{R}^n\times\mathbb{R}\times[0,\infty)\rightarrow\mathbb{R}$ is the disturbance of the system that is generally considered unknown in this study, $g(\cdot):[0,\infty)\rightarrow\mathbb{R}$ is a nonzero known input gain, and $\gamma$ is an unknown constant. 

\begin{assumption}
\label{ass:1}
For system \eqref{eq:4}, there exist $\phi,\phi_0\geq0$ and $\gamma_{\mathrm{min}}>0$ such that $|f(x,u,t)|\leq\phi\|x\|+\phi_0$ and $\gamma_{\mathrm{min}}\leq\gamma$.
\end{assumption}

The following theorem states the main result of this letter, which is proved in Section \ref{sec:proof}.

\begin{theorem}
\label{th:1}
Suppose Assumption \ref{ass:1} is satisfied for some $\phi$, $\phi_0$, and $\gamma_{\mathrm{min}}$. Let function $\pi(\cdot,\cdot,\cdot):\mathbb{R}^n\times[0,\infty)\times(0,\infty)\rightarrow\mathbb{R}$ be defined as
\begin{equation}
\label{eq:5}
\begin{split}
\pi(x,t,\tau)&=\sum_{j=1}^n\sum_{i=1}^{j}\stirlingii{j-1}{i-1}\frac{c_j}{\alpha^{n-j+1}}\frac{(-1)^{j-i}}{(\tau-t)^{n-i+1}}x_i\\
&-\sum_{j=2}^{n}\stirlingii{n}{j-1}\frac{(-1)^{n-j+1}}{(\tau-t)^{n-j+1}}x_j
\end{split}
\end{equation}
with parameters $c_j$, $j=1,\cdots,n$ selected such that the following matrix is Hurwitz:
\begin{equation}
\label{eq:6}
E=\left[\begin{array}{ccccc}
0 & 1 & 0 & \cdots & 0  \\
0 & 0 & 1 & \cdots & 0  \\
\vdots & \vdots & \vdots & \ddots & \vdots  \\
0 & 0 & 0 & \cdots & 1  \\
c_1 & c_2 & c_3 & \cdots & c_n 
\end{array}\right]\in\mathbb{R}^n
\end{equation}
and let $P\in\mathbb{R}^{n\times n}$ be the solution of the Lyapunov equation $E^TP+PE+2\mathbb{I}_n=0$. Then, the closed-loop solution of system \eqref{eq:4} under a PTC $u=\pi(x,t,\tau)/(\gamma_{\mathrm{min}}g(t))$ is:
\begin{enumerate}
\item globally triangularly attractive, if $\alpha$ is selected such that:
\begin{equation}
\label{eq:a1}
\begin{split}
\alpha<\min\left\{\frac{\lambda_{\mathrm{min}}(P)}{n\lambda_{\mathrm{max}}(P)\lambda_{\mathrm{min}}(P)+n!\lambda_{\mathrm{max}}^2(P)},\frac{1}{\tau}\right\}
\end{split}
\end{equation}
\item globally triangularly stable, if Assumption \ref{ass:1} is satisfied with $\phi_0=0$ and $\alpha$ is selected small enough such that in addition to \eqref{eq:a1}, it satisfies:
\begin{equation}
\label{eq:a2}
\alpha\leq\frac{\lambda_{\mathrm{min}}(P)}{\lambda_{\mathrm{min}}(P)+n!\lambda_{\mathrm{max}}^2(P)\left(\tau\phi+1\right)}
\end{equation}
\end{enumerate}
\end{theorem}

\begin{remark}
\label{rem:3}
It is provable that a controller, written in terms of gains multiplied by state variables, cannot reach zero equilibrium within a (known or unknown) finite time unless the gains approach infinity as the time approaches the convergence moment. This fact is also seen in the proposed PTC of Theorem \ref{th:1}. However, in a prescribed-time scheme, the convergence time and the singularity moment are known. In addition, the system state can reach any nonzero error by finite values of PTC gains. Therefore, to avoid singularity problems in practice, the termination time can be set slightly before $t=\tau$, depending on the required tolerance and the processor memory.
\end{remark}

\begin{remark}
\label{rem:4}
The user is free in selecting the value of $\tau$. However, it can be seen from the proposed PTC that the value of control input at $t=0$ is $\pi(x(0),0,\tau)=\sum_{j=1}^n\sum_{i=1}^j\stirlingii{j-1}{i-1}c_j[\alpha^{j-n-1}(-1)^{j-i}/\tau^{n-i+1}]x_i(0)-\sum_{j=2}^n\stirlingii{n}{j-1}[1/(-\tau)^{n-j+1}]x_j(0)$, which means that the initial control input increases by decreasing the convergence time $\tau$.
\end{remark}

Two examples are presented in the following to visualize the behavior of some systems under the proposed PTC\footnote{MATLAB\textsuperscript{\tiny\textregistered} codes and Simulink\textsuperscript{\tiny\textregistered} models for the proposed controller can be found in \href{https://github.com/a-shakouri/prescribed-time-control}{https://github.com/a-shakouri/prescribed-time-control}}.

\begin{example}
\label{ex:2}
Consider a second-order system as \eqref{eq:4} with $n=2$ and assume $g(t)=1$. Suppose that the disturbance term is also a function of control input as $f(x,u,t)=50\cos(u)+\cos(t)x_1+\exp(\sin(x_1))x_2$ for which Assumption \ref{ass:1} holds with $\phi=e=2.71828$ and $\phi_0=50$. Also, we assume $\gamma=1.1$ but it is unknown for the controller and it is only known that $\gamma\geq\gamma_{\mathrm{min}}=1$. Since $\phi_0\neq0$, the PTC can only guarantee triangular attractivity, for which parameter $\alpha$ is selected independent from the values of $\phi$ and $\phi_0$. We select $c_1=-1$, $c_2=-2$ and $\alpha=0.0214$ to satisfy condition \eqref{eq:a1} for $\tau=10$, $15$, and $20$. Fig. \ref{fig:3} shows the system response starting from $x(0)=[10,10]^T$.
\end{example}

\begin{figure}[!h]
\centering\includegraphics[width=1\linewidth]{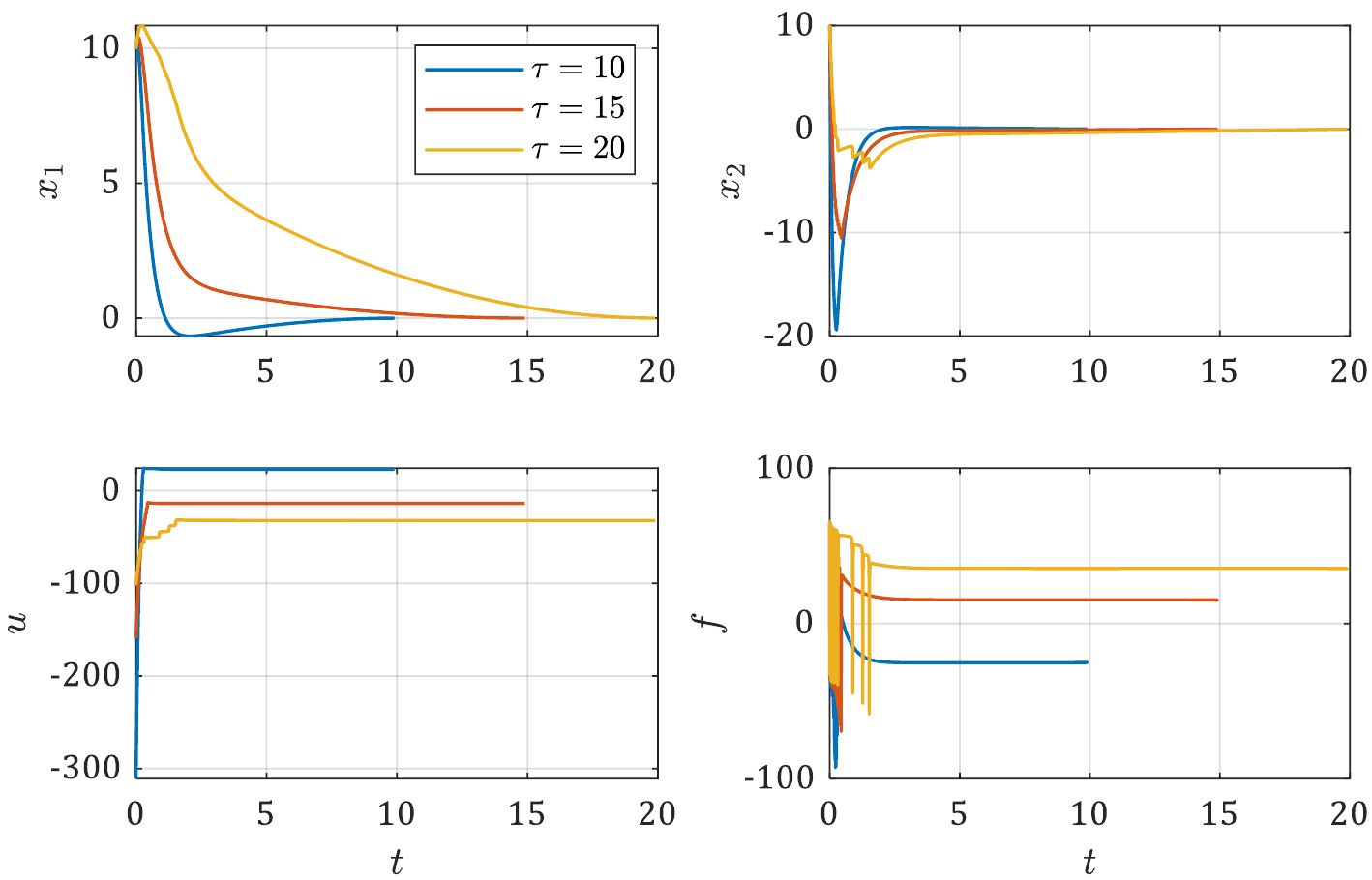}
\caption{Simulation results for the second-order system of Example \ref{ex:2} under the proposed PTC.}
\label{fig:3}
\end{figure}

\begin{table}[h!]
\begin{center}
\caption{PTC for a 4th-order system: $\pi(x,t,\tau)=\sum_{i+1}^4p_i(t,\tau)x_i$.}
\label{table:1}
\begin{tabular}{l|l}
\hline
$i$ & $p_i(t,\tau)$\\
\hline
 $1$ & $\frac{c_1}{\alpha^{4}(\tau-t)^{4}}$\\
 $2$ & $\frac{c_2}{\alpha^{3}(\tau-t)^{3}}-\frac{c_3}{\alpha^{2}(\tau-t)^{3}}+\frac{c_4}{\alpha(\tau-t)^{3}}+\frac{1}{(\tau-t)^{3}}$\\
 $3$ & $\frac{c_3}{\alpha^{2}(\tau-t)^{2}}-\frac{3c_4}{\alpha(\tau-t)^{2}}-\frac{7}{(\tau-t)^{2}}$\\
 $4$ & $\frac{c_4}{\alpha(\tau-t)}+\frac{6}{(\tau-t)}$\\
 \hline
\end{tabular}
\end{center}
\end{table}

\begin{example}
\label{ex:3}
Consider a fourth-order system as \eqref{eq:4} with $n=4$ and assume $g(t)=1$. For this system, the PTC expressed by \eqref{eq:5} can be written as $u=\pi(x,t,\tau)=\sum_{i+1}^4p_i(t,\tau)x_i$ with time-dependent gains stated in Table \ref{table:1}. Suppose that the disturbance is $f(x)=\sum_{i=1}^4w_ix_i$ where $w_i\sim \mathcal{U}(-10^{-3},10^{-3})$ for which Assumption \ref{ass:1} holds with $\phi=10^{-3}$. The input gain is unity, hence $\gamma_{\mathrm{min}}=1$. Since $\phi_0=0$, the triangular stability is achievable. For a convergence time of $\tau=10$, selecting $c_1=-1$, $c_2=c_4=-4$, and $c_3=-6$, conditions \eqref{eq:a2} and \eqref{eq:a1} are satisfied with $\alpha=1.6810\times10^{-5}$. The simulation results are plotted in Fig. \ref{fig:4} for $x(0)=[10,10,10,10]^T$.
\end{example}

\begin{figure}[!h]
\centering\includegraphics[width=1\linewidth]{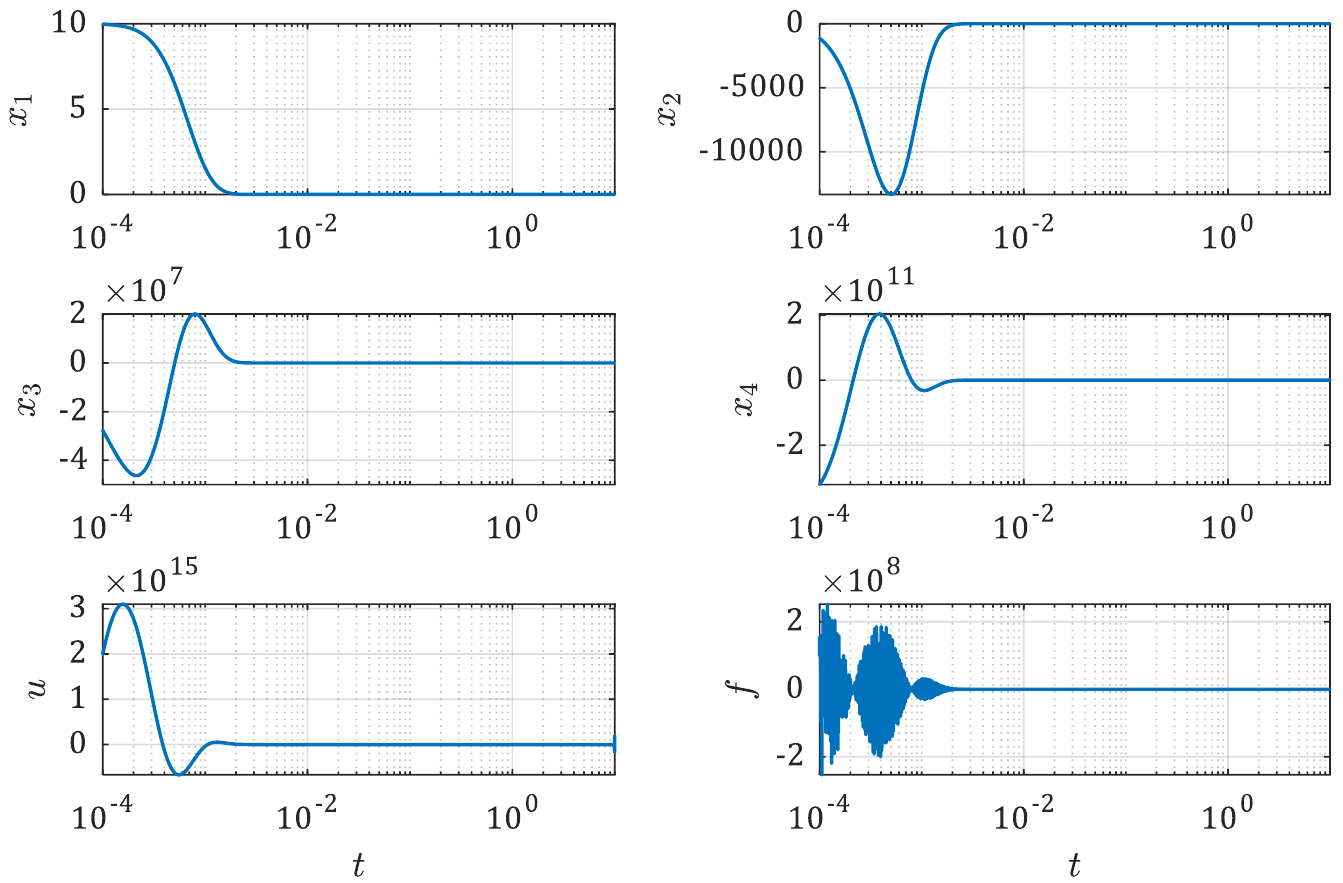}
\caption{Simulation results for the fourth-order system of Example \ref{ex:3} under the proposed PTC.}
\label{fig:4}
\end{figure}

\section{Proof of Theorem \ref{th:1}}
\label{sec:proof}

In this section, we explain how the PTC of Theorem \ref{th:1} is obtained. First, we discuss the unperturbed system with (known) unity gain. Next, the unperturbed system with uncertain input gain is studied. Finally, the boundedness of the controller is proved. Consider the following definition:

\begin{definition}
\label{def:kappa}
Define the following classes of functions:
\begin{enumerate}
\item A continuous function $\kappa(\cdot):[0,\tau)\rightarrow[0,\infty)$ is said to be class $\mathcal{K}$ (or $\kappa\in\mathcal{K}(\tau)$) if it is strictly increasing subject to $\lim_{t\rightarrow0^+}\kappa(t)=0$ and $\lim_{t\rightarrow\tau^-}\kappa(t)=\infty$ \cite{khalil2002nonlinear}.
\item A continuous function $\mu(\cdot):[0,\infty)\rightarrow[0,\tau)$ is said to be class $\mathcal{M}$ (or $\mu\in\mathcal{M}(\tau)$) if its inverse function is class $\mathcal{K}$ (or $\mu^{-1}\in\mathcal{K}(\tau)$). Therefore, $\mu$ is a continuous increasing function subject to $\lim_{t\rightarrow0^+}\mu(t)=0$ and $\lim_{t\rightarrow\infty}\mu(t)=\tau$. 
\end{enumerate}
\end{definition}

\subsection{State boundedness of the unperturbed system}

Consider a chain of integrators in terms of state vector $\xi\in\mathbb{R}^n$, as an auxiliary system adopted for the design of PTC, with $u=\sum_{i=i}^nc_i\xi_i$ as an  input:
\begin{equation}
\label{eq:IV_1}
\left\{\begin{array}{lcl}
\dot{\xi}_i & = & \xi_{i+1}; \hspace{2mm} i=1:n-1  \\
\dot{\xi}_n & = & \sum_{i=1}^nc_i\xi_i
\end{array}\right.
\end{equation}
which is exponentially stable if the matrix $E$ stated in \eqref{eq:6} is Hurwitz. To find a trajectory bound, consider a Lyapunov function $V=\xi^TP\xi$ where $P$ is the solution of the Lyapunov equation $E^TP+PE+2\mathbb{I}_n=0$ which results in $\dot{V}=-2\|\xi\|^2$. We have $\lambda_{\mathrm{min}}(P)\|\xi\|^2\leq V\leq\lambda_{\mathrm{max}}(P)\|\xi\|^2$, hence, $\dot{V}\leq-2V/\lambda_{\mathrm{max}}(P)$. Applying the differential form of Gronwall's inequality (see Theorem 1.9.1 in \cite{lakshmikantham1969differential}) leads to $V(t)\leq V(0)\exp\left(-2t/\lambda_{\mathrm{max}}(P)\right)$. Therefore, we have:
\begin{equation}
\label{eq:IV_1ppp}
\begin{split}
\|\xi(t)\|\leq \sqrt{\frac{\lambda_{\mathrm{max}}(P)}{\lambda_{\mathrm{min}}(P)}}\|\xi(0)\|\exp\left(-t/\lambda_{\mathrm{max}}(P)\right)
\end{split}
\end{equation}

Let $\acute{\kappa}(\mu)\coloneqq d\kappa(\mu)/d\mu$ and $\dot{\mu}(t)\coloneqq d\mu(t)/dt$. First, we develope a PTC for system \eqref{eq:IV_1}. Next, we show that the obtained PTC also works for system \eqref{eq:4} where $f\neq0$. Accordingly, we map system \eqref{eq:IV_1} from $[0,\infty)$ to a prescribed-time interval $[0,\tau)$ via $\mu(t)\in\mathcal{M}(\tau)$ to obtain a new system with state $y=[y_1,\cdots,y_n]^T\in\mathbb{R}^n$ using the following equivalent rules:
\begin{equation}
\label{eq:IV_2}
y_1(\mu(t))=\xi_1(t)
\end{equation}
\begin{equation}
\label{eq:IV_2p}
\xi_1(\kappa(\mu))=y_1(\mu)
\end{equation}
that \eqref{eq:IV_2} can be solved for the rest of state variables as:
\begin{equation}
\label{eq:IV_3}
\begin{array}{rcl}
\dot{\mu}(t)y_2(\mu(t))&=&\xi_2(t) \\
\mu^{(2)}(t)y_2(\mu(t))+\dot{\mu}^2(t)y_3(\mu(t))&=&\xi_3(t) \\
&\smash{\vdots}& 
\end{array}
\end{equation}
where $y_{i+1}=dy_i/d\mu$, and \eqref{eq:IV_2p} can similarly yield the followings:
\begin{equation}
\label{eq:IV_3p}
\begin{array}{rcl}
\acute{\kappa}(\mu)\xi_2(\kappa(\mu))&=&y_2(\mu) \\
{\kappa}^{(2)}(\mu)\xi_2(\kappa(\mu))+\acute{\kappa}^2(\mu)\xi_3(\kappa(\mu))&=&y_3(\mu) \\
&\smash{\vdots}&
\end{array}
\end{equation}
We define $\mu(t)=\tau(1-\exp(-\alpha t))$ and observe that $\dot{\mu}(t)=\alpha\tau\exp(-\alpha t)$, therefore, $\mu^{(i)}(t)=(-\alpha)^{i-1}\dot{\mu}(t)$ for all $i\in\mathbb{N}$. According to this assumption, we have
$\kappa(\mu)=-(1/\alpha)\ln(1-\mu/\tau)$, therefore $\acute{\kappa}(\mu)=(1/\alpha)/(\tau-\mu)$ and  $\kappa^{(i)}(\mu)=(-\alpha)^{i-1}(i-1)!\acute{\kappa}^{i}(\mu)$. Hence, when this class $\mathcal{M}$ function is used, procedures \eqref{eq:IV_3} and \eqref{eq:IV_3p} can be expressed, respectively, as:
\begin{equation}
\label{eq:IV_4}
\xi_j(t)=\sum_{i=2}^{j}\stirlingii{j-1}{i-1}(-\alpha)^{j-i}\dot{\mu}^{i-1}(t)y_i(\mu),\hspace{2mm}j\geq2
\end{equation}
\begin{equation}
\label{eq:IV_4p}
y_j(\mu)=\acute{\kappa}^{j-1}(\mu)\sum_{i=2}^{j}\stirlingi{j-1}{i-1}(-\alpha)^{j-i}\xi_i(t),\hspace{2mm}j\geq2
\end{equation}
Also, the following statement is obtainable for $\dot{\xi}_{n}\equiv \xi_{n+1}$:
\begin{equation}
\label{eq:IV_5}
\dot{\xi}_{n}(t)=\frac{\acute{y}_n(\mu)}{\acute{\kappa}^{n}(\mu)}+\sum_{i=2}^{n}\stirlingii{n}{i-1}(-\alpha)^{n-i+1}\frac{y_i(\mu)}{\acute{\kappa}^{i-1}(\mu)}
\end{equation}
Substitute \eqref{eq:IV_5} into \eqref{eq:IV_1} and yields the following system:
\begin{equation}
\label{eq:IV_6}
\left\{\begin{array}{lcl}
\acute{y}_i & = & y_{i+1}; \hspace{2mm} i=1:n-1  \\
\acute{y}_n & = & -\sum_{i=2}^{n}\stirlingii{n}{i-1}(-\alpha)^{n-i+1}\acute{\kappa}^{n-i+1}y_i+\acute{\kappa}^{n}c_1y_1 \\
&&+ \sum_{j=2}^nc_j\sum_{i=2}^{j}\stirlingii{j-1}{i-1}(-\alpha)^{j-i}\acute{\kappa}^{n-i+1}y_i
\end{array}\right.
\end{equation}
System \eqref{eq:IV_6} mimics the behavior of system \eqref{eq:IV_1} in a finite interval $[0,\tau)$. Substitute $\acute{\kappa}$ in \eqref{eq:IV_6} to yield:
\begin{equation}
\label{eq:IV_7}
\left\{\begin{array}{lcl}
\acute{y}_i & = & y_{i+1}; \hspace{2mm} i=1:n-1  \\
\acute{y}_n & = & -\sum_{i=2}^{n}\stirlingii{n}{i-1}\frac{(-1)^{n-i+1}}{(\tau-\mu)^{n-i+1}}y_i+\frac{c_1}{\alpha^{n}(\tau-\mu)^{n}}y_1 \\
&&+ \sum_{j=2}^n\frac{c_j}{\alpha^{n-j+1}}\sum_{i=2}^{j}\stirlingii{j-1}{i-1}\frac{(-1)^{j-i}}{(\tau-\mu)^{n-i+1}}y_i
\end{array}\right.
\end{equation}
It can be verified that a controller $u=\pi(y,\mu,\tau)$, as proposed by \eqref{eq:5}, forces the system to behave as \eqref{eq:IV_7}. It is provable that identities \eqref{eq:IV_4} and \eqref{eq:IV_4p} can be represented in the following matrix forms:
\begin{equation}
\label{eq:IV_8}
\xi(t)=(A_n\circ(S_nK_n^{-1}))y(\mu)
\end{equation}
\begin{equation}
\label{eq:IV_9}
y(\mu)=(A_n\circ(M_n^{-1}s_n))\xi(t)
\end{equation}
where $A_n\in\mathbb{R}^{n\times n}$ is a $\tau$-dependent lower triangular Toeplitz matrix as:
\begin{equation}
\label{eq:7}
A_n(i,j)=\left\{\begin{array}{lcl}
(-\alpha)^{i-j} & \mathrm{if} & i\geq j \\
0 & \mathrm{if} & \mathrm{otherwise}
\end{array}\right.
\end{equation}
and $K_n$ and $M_n$ are diagonal matrices as:
\begin{equation}
\label{eq:IV_10}
K_n=\mathrm{diag}\left(1,\acute{\kappa},\acute{\kappa}^2,\cdots,\acute{\kappa}^{n-1}\right)
\end{equation}
\begin{equation}
\label{eq:IV_11}
M_n=\mathrm{diag}\left(1,\dot{\mu},\dot{\mu}^2,\cdots,\dot{\mu}^{n-1}\right)
\end{equation}

We need the following lemma to construct an inequality for the mapped system:
\begin{lemma}
\label{lem:2}
If $C=A\circ B$, then we have $\|C\|\leq\|A\|\|B\|$.
\end{lemma}
\begin{proof}
Observe that $\|A\|\|B\|=\|A\otimes B\|$ and consider the fact that the Hadamard product is a principal submatrix of the Kronecker product.
\end{proof}

Given \eqref{eq:IV_8}, for the initial condition one can obtain $\|\xi(0)\|\leq\|A_n\|\|S_n\|\|y(0)\|$. According to \eqref{eq:IV_9} and Lemma \ref{lem:2}, assuming $\alpha\leq1/\tau$, we can write the following inequality:
\begin{equation}
\label{eq:IV_13}
\resizebox{\columnwidth}{!}{
$\begin{split}
\|y(\mu)\|&\leq\|A_n\circ(M_n^{-1}s_n)\|\|\xi(t)\|\leq\|A_n\|\|M_n^{-1}\|\|s_n\|\|\xi(t)\|\\
&\leq\|A_n\|\|s_n\|\frac{\|\xi(t)\|}{\dot{\mu}^{n-1}}=\frac{\|A_n\|\|s_n\|}{(\alpha\tau)^{n-1}}\frac{\|\xi(t)\|}{\exp(\alpha(1-n)t)}\\
&\leq\frac{\|A_n\|^2\|s_n\|\|S_n\|\|y(0)\|}{{(\alpha\tau)^{n-1}}}\sqrt{\frac{\lambda_{\mathrm{max}}(P)}{\lambda_{\mathrm{min}}(P)}}\\
&\cdot\exp\left[\left(\alpha(n-1)-1/\lambda_{\mathrm{max}}(P)\right)t\right]
\end{split}$}
\end{equation}
Substitute $t=\kappa(\mu)=-(1/\alpha)\ln(1-\mu/\tau)$ to obtain:
\begin{equation}
\label{eq:IV_14}
\resizebox{\columnwidth}{!}{
$\begin{split}
\|y(\mu)\|&\leq\frac{\|A_n\|^2\|s_n\|\|S_n\|\|y(0)\|}{(\alpha\tau)^{n-1}}\sqrt{\frac{\lambda_{\mathrm{max}}(P)}{\lambda_{\mathrm{min}}(P)}}\\
&\cdot\exp\left[\left(1-n+\frac{1}{\alpha\lambda_{\mathrm{max}}(P)}\right)\ln\left(1-\frac{\mu}{\tau}\right)\right]\\
&=\frac{\|A_n\|^2\|s_n\|\|S_n\|\|y(0)\|}{(\alpha\tau)^{n-1}}\sqrt{\frac{\lambda_{\mathrm{max}}(P)}{\lambda_{\mathrm{min}}(P)}}\left(1-\frac{\mu}{\tau}\right)^{1-n+\frac{1}{\alpha\lambda_{\mathrm{max}}(P)}}
\end{split}$}
\end{equation}
that satisfies the triangular stability condition if the power of $\Lambda(\mu/\tau)=1-\mu/\tau$ is greater than or equal to unity, which is condition \eqref{eq:a1} for $\gamma=1$. According to the obtained inequalities, the mapped nominal system is triangularly stable as long as the original system is exponentially stable. Note that the closed-loop system of an unperturbed chain of integrators under the proposed PTC, when $x_i$ is substituted by $y_i$, and $t$ is substituted by $\mu$, has the same behavior as the mapped system. 

\begin{remark}
\label{rem:6}
The matrix norms used in this section can be further simplified and replaced by their upper bounds. For instance, the following inequalities hold:
\begin{enumerate}
\item $\|s_n\|\leq \sqrt{n}(n-1)!$
\item $\|S_n\|\leq \sqrt{n} B_{n-1}$, where $B_i$ is the $i$th Bell number
\item $\|A_n\|\leq \sqrt{n}\sum_{i=0}^{n-1}\alpha^i\Rightarrow\|A_n\|\leq \sqrt{n}/(1-\alpha)$ if $\alpha<1$
\end{enumerate}
\end{remark}

\subsection{State boundedness of the perturbed system}

At this step, we have analyzed the behavior of the perturbed system with uncertain input gain under the proposed controller. Without loss of generality, assume $g(t)=1$, and inverse-map system \eqref{eq:1} under the proposed controller of Theorem \ref{th:1} to yield the following system:
\begin{equation}
\label{eq:IV_16}
\left\{\begin{array}{lcl}
\dot{\xi}_i & = & \xi_{i+1}; \hspace{2mm} i=1:n-1  \\
\dot{\xi}_n & = & (1-\frac{1}{\rho})\sum_{i=2}^{n+1}\stirlingi{n}{i-1}(-\alpha)^{n+1-i}\xi_i(t) \\
&&+\sum_{i=i}^nc_i\xi_i+\frac{\dot{\mu}^{n}}{\rho}f[(A_n\circ(M_n^{-1}s_n))\xi(t),t]
\end{array}\right.
\end{equation}
where $\rho=\gamma/\gamma_{\mathrm{min}}\geq1$. System \eqref{eq:IV_16} can be written as:
\begin{equation}
\label{eq:IV_17}
\dot{\xi}=E\xi+\left(1-1/\rho\right)s(\xi,t)+(1/\rho)\bar{f}(\xi,t)
\end{equation}
where $\bar{f}(\xi,t)=[0,\cdots,0,\dot{\mu}^{n}f[(A_n\circ(M_n^{-1}s_n))\xi(t),t]]^T$ and $s(\xi,t)=[0,\cdots,0,\sum_{i=2}^{n+1}\stirlingi{n}{i-1}(-\alpha)^{n+1-i}\xi_i(t)]^T$. According to Assumption \ref{ass:1}, we have $\|\bar{f}(\xi,t)\|=\dot{\mu}^{n}\|f[(A_n\circ(M_n^{-1}s_n))\xi(t),t]\|\leq\dot{\mu}^{n}\phi\|A_n\circ(M_n^{-1}s_n)\|\|\xi(t)\|+\dot{\mu}^{n}\phi_0\leq\dot{\mu}\phi\|A_n\|\|s_n\|\|\xi(t)\|+\dot{\mu}^{n}\phi_0$. Also, one can verify that $\|s(\xi,t)\|\leq n!\alpha\|\xi(t)\|$. To obtain the trajectory bounds, consider a Lyapunov function as $V=\xi^TP\xi$ where $P$ is the solution of the Lyapunov equation $E^TP+PE+2\mathbb{I}_n=0$. Thus, we have $\dot{V}=-2\|\xi\|^2+2(1-1/\rho)\xi^TPs(\xi,t)+2(1/\rho)\xi^TP\bar{f}(\xi,t)\leq-2\|\xi\|^2+2(1-1/\rho)n!\alpha\lambda_{\mathrm{max}}(P)\|\xi\|^2+2(1/\rho)\dot{\mu}\lambda_{\mathrm{max}}(P)\phi\|A_n\|\|s_n\|\|\xi\|^2+2(1/\rho)\dot{\mu}^{n}\lambda_{\mathrm{max}}(P)\phi_0\|\xi\|$ that reduces to $\dot{V}\leq2\alpha\theta\exp(-\alpha t)V+2\vartheta V+2\alpha^n\tau^n\phi_0[\lambda_{\mathrm{max}}(P)/\sqrt{\lambda_{\mathrm{min}}(P)}]\exp(-n\alpha t)\sqrt{V}$ where $\theta=(1/\rho)\tau\phi\|A_n\|\|s_n\|\lambda_{\mathrm{max}}(P)/\lambda_{\mathrm{min}}(P)$ and $\vartheta=(1-1/\rho)n!\alpha\lambda_{\mathrm{max}}(P)/\lambda_{\mathrm{min}}(P)-1/\lambda_{\mathrm{max}}(P)$. Changing the variable as $V=W^2$, the obtained inequality can be stated as
\begin{equation}
\label{eq:IV_21p1}
\resizebox{\columnwidth}{!}{
$\begin{split}
\dot{W}&\leq\left[\alpha\theta\exp(-\alpha t)+\vartheta\right]W+\alpha^n\tau^n\phi_0\frac{\lambda_{\mathrm{max}}(P)}{\sqrt{\lambda_{\mathrm{min}}(P)}}\exp(-n\alpha t)
\end{split}$}
\end{equation}
According to the solution of the nonhomogeneous linear differential equation (the well known variation of constants formula) for the right-hand side of \eqref{eq:IV_21p1} and the Petrovitsch's theorem of differential inequalities (or Theorem 1.2.3 in \cite{lakshmikantham1969differential}), there exists $t_0>0$ such that for all $t\in[t_0,\infty)$ the following inequality holds:
\begin{equation}
\label{eq:IV_21p2}
\resizebox{0.85\columnwidth}{!}{$\begin{split}
W(t)&\leq W(0)\exp\left[-\theta\exp(-\alpha t)+\vartheta t+\theta\right]\\
&+\frac{\alpha^n\tau^n\phi_0\lambda_{\mathrm{max}}(P)}{\sqrt{\lambda_{\mathrm{min}}(P)}}\exp\left[-\theta\exp(-\alpha t)+\vartheta t\right]\\
&\cdot\int_0^t\exp\left[\theta\exp(-\alpha s)-(\vartheta+n\alpha)s\right]ds
\end{split}$}
\end{equation}
After calculating the integral (by the Taylor series expansion of the exponential function) and using $\exp[-\theta\exp(-\alpha t)]\leq1$, equation \eqref{eq:IV_21p2} reduces to the following inequality:
\begin{equation}
\label{eq:IV_21p3}
\resizebox{0.85\columnwidth}{!}{
$\begin{split}
W(t)&\leq W(0)\exp(\vartheta t+\theta)-\frac{\alpha^n\tau^n\phi_0\lambda_{\mathrm{max}}(P)}{\sqrt{\lambda_{\mathrm{min}}(P)}}\exp(\vartheta t)\\
&\cdot\sum_{k=0}^\infty\frac{(\theta^k/k!)(\exp[(-(k+n)\alpha-\vartheta)t]-1)}{\alpha(k+n)+\vartheta}
\end{split}$}
\end{equation}
Since we have $\sqrt{\lambda_{\mathrm{min}}(P)}\|\xi\|\leq W\leq\sqrt{\lambda_{\mathrm{max}}(P)}\|\xi\|$, then:
\begin{equation}
\label{eq:IV_22}
\resizebox{0.85\columnwidth}{!}{
$\begin{split}
\|\xi(t)\|&\leq \sqrt{\frac{\lambda_{\mathrm{max}}(P)}{\lambda_{\mathrm{min}}(P)}}\|\xi(0)\|\exp(\vartheta t+\theta)\\
&-\frac{\alpha^n\tau^n\phi_0\lambda_{\mathrm{max}}(P)}{\lambda_{\mathrm{min}}(P)}\sum_{k=0}^\infty\frac{(\theta^k/k!)\exp[-(k+n)\alpha t]}{\alpha(k+n)+\vartheta}\\
&+\frac{\alpha^n\tau^n\phi_0\lambda_{\mathrm{max}}(P)}{\lambda_{\mathrm{min}}(P)}\sum_{k=0}^\infty\frac{(\theta^k/k!)\exp(\vartheta t)}{\alpha(k+n)+\vartheta}
\end{split}$}
\end{equation}
Note that an upper bound exist for each sum $\sum_{k=0}^\infty(\theta^k/k!)\exp[-\alpha kt]/[\alpha(k+n)+\vartheta]\leq\exp(\theta)/(\alpha n+\vartheta)$ and $\sum_{k=0}^\infty(\theta^k/k!)/[\alpha(k+n)+\vartheta]\leq\exp(\theta)/(\alpha n+\vartheta)$. Hence, similar to \eqref{eq:IV_13}, one can obtain the following (the obtained upper bounds for the sums are also substituted):
\begin{equation}
\label{eq:IV_23}
\resizebox{0.85\columnwidth}{!}{
$\begin{split}
\|y(\mu)\|&\leq\frac{\|A_n\|^2\|s_n\|\|S_n\|\|y(0)\|\exp(\theta)}{\alpha^{n-1}\tau^{n-1}}\sqrt{\frac{\lambda_{\mathrm{max}}(P)}{\lambda_{\mathrm{min}}(P)}}\\
&\cdot\exp\left[(\vartheta+\alpha(n-1))t\right]\\
&+\frac{\|A_n\|\|s_n\|\alpha\tau\phi_0\exp(\theta)}{\alpha n+\vartheta}\frac{\lambda_{\mathrm{max}}(P)}{\lambda_{\mathrm{min}}(P)}\\
&\cdot(\exp\left[(\vartheta+\alpha(n-1))t\right]-\exp(-\alpha t))
\end{split}$}
\end{equation}
Thus, substituting $t=\kappa(\mu)=-(1/\alpha)\ln(1-\mu/\tau)$, it can be verified that the following inequality holds for $\mu\in[\mu(t_0),\infty)$:
\begin{equation}
\label{eq:IV_24}
\resizebox{0.85\columnwidth}{!}{
$\begin{split}
\|y(\mu)\|&\leq\frac{\|A_n\|^2\|s_n\|\|S_n\|\|y(0)\|\exp(\theta)}{\alpha^{n-1}\tau^{n-1}}\sqrt{\frac{\lambda_{\mathrm{max}}(P)}{\lambda_{\mathrm{min}}(P)}}\\
&\cdot\left(1-\frac{\mu}{\tau}\right)^{1-\vartheta/\alpha-n}\\
&+\frac{\|A_n\|\|s_n\|\alpha\tau\phi_0\exp(\theta)}{\alpha n+\vartheta}\frac{\lambda_{\mathrm{max}}(P)}{\lambda_{\mathrm{min}}(P)}\\
&\cdot\left[\left(1-\frac{\mu}{\tau}\right)^{1-\vartheta/\alpha-n}-\left(1-\frac{\mu}{\tau}\right)\right]
\end{split}$}
\end{equation}
Therefore, state $y$ is triangularly attractive in terms of $\mu$ if all of the powers of $\Lambda(\mu/\tau)=1-\mu/\tau$ are greater than or equal to unity, which needs condition \eqref{eq:a1}.

In the case of $\phi_0=0$, we have $\dot{V}\leq2(\alpha\theta+\vartheta)V$. Applying the differential form of the Gronwall's inequality (or Theorem 1.9.1 in \cite{lakshmikantham1969differential}), results in $V(t)\leq V(0)\exp[2(\alpha\theta+\vartheta)t]$. Similar to the proof of the unperturbed case, one can find $\alpha\theta+\vartheta\leq0$ as well as \eqref{eq:a1} are sufficient conditions for the closed-loop system to be triangularly stable, which can be reduced to condition \eqref{eq:a2} after applying inequalities discussed in Remark \ref{rem:6}. As the final step, consider the fact that the closed-loop perturbed system under the proposed PTC, when $\xi_i$ is substituted by $y_i$ and $t$ is substituted by $\mu$, acts exactly as the mapped system.

\subsection{Controller boundedness of the perturbed system}

In this subsection, We are going to prove that the PTC proposed by Theorem \ref{th:1} produces finite values of control input. In terms of the mapped system variables (where $y_i$ and $\mu$ are used instead of $x_i$ and $t$), we know that $\pi(y,\mu,\tau)=\acute{y}_n(\mu)-f(y,\mu)$. Using the triangle inequality for vector norms, one can obtain $\|\pi(y,\mu,\tau)\|\leq\|\acute{y}_n(\mu)\|+|f(y,\mu)|\leq\|\acute{y}_n(\mu)\|+\phi\|y(\mu)\|+\phi_0$. When $\mu\rightarrow\tau$, from the results of the previous section we have $\|y_n(\mu)\|\leq\|y(\mu)\|\leq\sigma\Lambda^k(\mu/\tau)$ for some $\sigma>0$ and $k>1$, thus $1\geq\lim_{\mu\rightarrow\tau}\|y_n(\mu)\|/[\sigma\Lambda^k(\mu/\tau)]$. According to the L'H\^{o}pital's rule, $1\geq\lim_{\mu\rightarrow\tau}\|y_n(\mu)\|/[\sigma\Lambda^k(\mu/\tau)]=\lim_{\mu\rightarrow\tau}[d\|y_n(\mu)\|/d\mu]/[-k\sigma\Lambda^{k-1}(\mu/\tau)/\tau]$. Therefore, as $\mu\rightarrow\tau$, we have $\|\acute{y}_n(\mu)\|\leq d\|y_n(\mu)\|/d\mu\leq\sigma^\prime\Lambda^{k-1}(\mu/\tau)$ for some $\sigma^\prime>0$ while the power $k-1$ can be less than unity. By  a substitution, one can obtain that the proposed PTC is eventually bounded by  $\|\pi(y,\mu,\tau)\|\leq\sigma^{\prime\prime}\Pi(\mu/\tau)+\phi_0$ for some finite $\sigma^{\prime\prime}>0$ where $\Pi(\mu/\tau)$ is a rectangular function.

\section{Conclusions}

The notion of triangular stability has been defined in this letter, and it has been shown that a prescribed-time controller can provide global triangular stability (or attractivity) for a perturbed normal form system with uncertain input gain.
Future studies could investigate triangularly stable (or attractive) prescribed-time controllers for other types of systems (e.g., delayed systems, constrained systems, stochastic systems, multi-agent systems, etc.), under more severe disturbances. 

\bibliographystyle{IEEEtran}
\bibliography{root}

\end{document}